\documentclass[11pt]{amsart}

\usepackage{amsmath}
\usepackage{amssymb}
\usepackage[utf8]{inputenc}

\usepackage{float}
\usepackage{graphicx}
\usepackage{url}

\usepackage{color, colortbl}

\newtheorem{theo}{Theorem}[section]
\newtheorem{prop}{Proposition}[section]

\newtheorem{defi}{Definition}[section]

\newcommand{\proscal}[2]{\left\langle#1,#2\right\rangle}

\newcommand{\V}[1]{\mathrm{Var}\left[#1\right]}

\newcommand{\normal}[2]{\mathcal{N}\left(#1,\, #2\right)}

\newcommand{\R}{\mathbb{R}}
\newcommand{\Rp}{{\mathbb{R}^p}}
\newcommand{\Rd}{{\mathbb{R}^d}}

\newcommand{\Mg}{\mathcal{M}_g}

\newcommand{\colbf}[1]{{\mathbf{#1}}}
\newcommand{\vecbf}[1]{{\vec{\mathbf{#1}}}}
\newcommand{\tabbf}[1]{{\mathbf{#1}}}

\newcommand{\Id}{\mathrm{Id}}

\begin{document}

\title{Probabilistic Auto-Associative Models and Semi-Linear PCA}

\author{Serge Iovleff}
\date{\today}

\begin{abstract}
Auto-Associative models cover a large class of methods used in data analysis.
In this paper, we describe the generals properties of these models when the projection
component is linear and we propose and test an easy to implement Probabilistic Semi-Linear
Auto-Associative model in a Gaussian setting. We show it is a generalization of the PCA
model to the semi-linear case. Numerical experiments on simulated datasets and a real
astronomical application highlight the interest of this approach.
\end{abstract}

\maketitle

\section{Introduction}

Principal component analysis (PCA) \cite{Pearson, Hotel, Jolli} is a well
established tool for dimension reduction in multivariate
data analysis. It benefits from a simple geometrical interpretation.
Given a set of $n$ points $\tabbf{Y} = (\colbf{y}_1,\ldots,\colbf{y}_n)'$
with $\colbf{y}_i \in \Rp$ and an integer $0\leq d\leq p$,
PCA builds the $d$-dimensional affine subspace minimizing the
Euclidean distance to the scatter-plot \cite{Pearson}.
The application of principal component analysis postulates
implicitly some form of linearity. More precisely, one assumes
that the data cloud is directed, and that the data points can
be well approximated by there projections to the affine hyperplane
corresponding to the first $d$ principal components.

Starting from this point of view, many authors have proposed nonlinear extensions of this technique.
Principal curves or principal surfaces methods \cite{HAS89,Hastie2001,DEL01} belong to this family of
approaches, non-linear transformation of the original data set \cite{DUR93, BES95} too. The
auto-associative neural networks can also be view as a non-linear PCA model \cite{Baldi, Bei,
BISHOP95, Hinton97}. In \cite{GI} we propose the auto-associative models (AAM)
as candidates to the generalization of PCA using a projection pursuit
regression algorithm \cite{Fried81,Klinke} adapted to the auto-associative case.
A common point of these approaches is that they have the intent to estimate an
auto-associative model whose definition is given hereafter.
\begin{defi}\label{def:AAM}
A function $g$ is an auto-associative function
of dimension $d$ if it is a map from $\Rp$ to $\Rp$ that can be written
$g=R\circ P$ where $P$ (the ``Projection'') is a map from
$\Rp$ to $\Rd$ (generally $d<p$) and $R$ (the ``Restoration'' or the "Regression") is a map
from $\Rd$ to $\Rp$.

An auto-associative model (AAM) of dimension $d$ is a manifold $\mathcal{M}_g$ of the form
$$
\mathcal{M}_g = \{ \colbf{y}\in \R^p,\, \colbf{y}-g(\colbf{y}) = 0 \}
$$
where $g$ is an auto-associative function of dimension $d$.
\end{defi}

For example the PCA constructs an auto-associative model using
as auto-associative function an orthogonal projector on an affine subspace
of dimension $d$. More precisely we have
$$ g(\colbf{y}) = \colbf{m}+\sum_{i=1}^d \proscal{\vecbf{a}^i}{\colbf{y}-\colbf{m}}
   \vecbf{a}^i, \quad \colbf{y}\in \Rp
$$
with $\colbf{y}, \colbf{m}, \vecbf{a}_i \in \Rp$ and the vectors $\vecbf{a}_i$
chosen in order to maximize the projected variance. $g$ can be written $g=R\circ P$ with
$$P(\colbf{y}) = \left(\proscal{\vecbf{a}^1}{\colbf{y}-\colbf{m}}, \ldots,
                       \proscal{\vecbf{a}^d}{\colbf{y}-\colbf{m}} \right)$$
and
$$R(\colbf{x}) = \colbf{m} + x_1 \vecbf{a}^1 + \ldots + x_d \vecbf{a}^d$$
with $\colbf{x}=(x_1,\ldots,x_d)^\prime$. The AAM is then the affine subspace
given by the following equation
$$
\Mg = \left\{ \colbf{y} \in \Rp;\, \colbf{y} - \colbf{m}
              - \sum_{i=1}^d \proscal{\vecbf{a}^i}{\colbf{y}-\colbf{m}}\vecbf{a}^i = 0
\right\}
$$
Interested reader can check that principal curves, principal surfaces, auto-associative
neural networks, kernel PCA \cite{Schoelkopf99},
ISOMAP \cite{Tenenbaum2000} and local linear embedding \cite{Roweis2000}
have also the intent to estimate an AAM.

In the PCA approach the projection and the restoration function are both linear. It is thus natural
to say that the PCA is a \emph{Linear Auto-Associative Model}.
In the general case, the manifold $\mathcal{M}_g$ set can be empty (i.e. the auto-associative function $g$
have no fixed point) or very complicated to describe. Our aim in this paper is to study from a theoretical
and practical point of view the properties of some Auto-Associative models in an intermediary situation
between the PCA model and the general case:
we will assume that the projection function is linear and let the regression function be arbitrary.
We call the resulting AAM the \emph{Semi-Linear Auto-Associative Models} (SLAAM).

Having restricted our study to the SLAAM, we have to give us some criteria to maximize. As we said previously,
the PCA tries to maximize the projected variance or, equivalently, to minimize the residual variance. Common AAM
approaches used also the squared reconstruction error as criteria, or more recently a penalized criteria
\cite{Hastie2001}. However as pointed out by M. E. Tipping and C. M. Bishop \cite{TB99}, one limiting
disadvantage of this approach is the absence of a probability density model and associated likelihood measure.
The presence of a probabilistic model is desirable as
\begin{itemize}
  \item the definition of a likelihood measure permits comparison between concurrent models
  and facilitates statistical testing,
  \item A single AAM may be extended to a mixture of such models,
  \item if a probabilistic AAM is used to model the class conditional densities in a classification
  problem, the posterior probabilities of class membership may be computed.
\end{itemize}
We propose thus a Gaussian generative model for the SLAAM and try to estimate it using a maximum likelihood
approach. In the general case we are faced with a difficult optimization problem and we cannot go further
without additional assumptions. It will appear clearly that if $\tabbf{P}$ is \emph{known} then the
estimation problem of a SLAAM is very close to an estimation problem in a regression context. There is however
some differences we will enlighten. In particular it will appear that in order to get tractable
maximum-likelihood estimates, we have to impose some restrictions to the noise.
We call the resulting model of all these assumptions/simplifications a Semi-Linear Principal Component Analysis.
It does not seem possible to add non-linearity to the PCA and get tractable likelihood estimate for $\tabbf{P}$.
But clearly, the assumption that $\tabbf{P}$ is known is too strong in practice. We propose thus to estimate
it in a separate step using either the PCA or a contiguity analysis \cite{Lebart00} by extending our
previous work on the Auto-Associative models \cite{GI}. Finally, even if $\tabbf{P}$ is assumed known it
remains to estimate the regression function $R$ which is a non-linear function from $\Rd$ to $\Rp$.
If $d>1$ and $p$ is moderately high the task become very complicated. Thus we simplify once more the
model and assume that $R$ is additive inspired by the Generalized Additive Model (GAM) approach \cite{Hastie90}.

In view of the experiments we have performed and we present there, it seems we obtain a practical
and simple model which generalizes in an understandable way the PCA model to the non-linear case.

The paper is organized as follows. Section \ref{sec:PAAM} introduces the Probabilistic
Semi-Linear Auto-Associative Models (PSLAAM) and relate them to the PCA and Probabilistic PCA models.
In  section \ref{sec:NLPCA} we present the Probabilistic Semi-Linear PCA models and the estimation of
theirs parameters conditionally to the knowledge of the projection matrix $\tabbf{P}$. Section
\ref{sec:contiguityAnalysis} is devoted to the determination of the projection matrix $\tabbf{P}$ using
contiguity analysis. Data sets and experiments are detailed in Section \ref{sec::apps} with a real
astronomical data set. Finally, some concluding remarks are proposed in Section \ref{sec::conc}.

\section{Semi-Linear Auto-Associatif Models (SLAAM)}
\label{sec:PAAM}

\subsection{Geometrical properties of the SLAAM}
\label{subsec:GeoSLAAM}

Let us first consider a general auto-associative model as given in the definition \ref{def:AAM}.
We have the following evident property
\begin{prop}\label{prop:id}
Let $H= \{P(\colbf{y});\, \colbf{y}\in\Mg \}\subset\Rd$. On $H$ the projection function and the
regression function verify
\begin{equation}\label{eq:id}
 P\circ R = \Id_d
\end{equation}
where $\Id_d$ denote the identity function of $\Rd$.
\end{prop}
\begin{proof}
Let $\colbf{y}\in \Mg$ and let $\colbf{x} = P(\colbf{y})$, then
$$ \colbf{x} = P(\colbf{y}) = P(g(\colbf{y})) = P(R(P(\colbf{y}))) = P(R(\colbf{x})).$$
\end{proof}

As a consequence, we have the following ``orthogonality'' property verified by an AAM when $P$
is an additive function
\begin{prop}
Let $V= \{P(\colbf{y});\, \colbf{y}\in\Rp \}$ and assume that the property (\ref{eq:id}) extend on
$V$,
let $\colbf{y}\in \Rp$, $\bar{\colbf{y}} = R(P(\colbf{y}))$ and $\bar{\colbf{\varepsilon}} =
\colbf{y} - \bar{\colbf{y}}$.
If $P$ is additive, i.e. $P(\colbf{y} + \colbf{y}') = P(\colbf{y}) + P(\colbf{y}')$, then
$$ P(\bar{\colbf{\varepsilon}}) = 0.$$
\end{prop}
\begin{proof}
Using the property (\ref{eq:id}), we have on one hand $P(\bar{\colbf{y}}) = P(R(P(\colbf{y}))) =
P(\colbf{y})$.
While on the other hand
$P(\bar{\colbf{y}}) = P(\colbf{y} -\bar{\colbf{\varepsilon}}) =  P(\colbf{y}) -
P(\bar{\colbf{\varepsilon}})$
giving the announced result.
\end{proof}
Clearly we have $H\subset V$ and the assumption given in this proposition seems quite natural.
We focus now on the semi-linear case and we assume that
\begin{equation}\label{eq:linear}
P(\colbf{y}) = \left( \proscal{\vecbf{a}^1}{\colbf{y}}, \ldots, \proscal{\vecbf{a}^d}{\colbf{y}}
\right) = \tabbf{P} \colbf{y}.
\end{equation}
with $\tabbf{P} = (\vecbf{a}^{1},\ldots, \vecbf{a}^{d})^\prime$ a matrix of size $(d,p)$.

\begin{prop}\label{prop:SLAAM}
Let $g = R\circ P$ be an auto-associative function, with $P$ given in (\ref{eq:linear}) and $R$
verifying the property (\ref{eq:id}). Let $\mathcal{B} = \left( \vecbf{a}^{1},\ldots,
\vecbf{a}^{d}, \vecbf{a}^{d+1},\ldots,\vecbf{a}^{p} \right)$ be an orthonormal basis of $\Rp$
with $(\vecbf{a}^{d+1},\ldots,\vecbf{a}^{p})$ chosen arbitrarily. Let $\colbf{y} \in \mathcal{M}_g$,
and let $\tilde{\colbf{y}}$ and $\tilde{r}$ denote respectively the vector $\colbf{y}$ and the
auto-associative function $r$ in the basis $\mathcal{B}$, then
\begin{equation}\label{eq:can}
\begin{pmatrix}
\tilde{y}_1 \\ \vdots \\ \tilde{y}_d\\ \tilde{y}_{d+1} \\ \vdots \\ \tilde{y}_p
\end{pmatrix}
=
\begin{pmatrix}
\tilde{y}_1 \\ \vdots \\ \tilde{y}_d \\ \tilde{r}_{d+1}(\tilde{y}_1, \ldots, \tilde{y}_d) \\ \vdots
\\
 \tilde{r}_p(\tilde{y}_1, \ldots, \tilde{y}_d)
\end{pmatrix}.
\end{equation}
\end{prop}
\begin{proof}
It suffices to notice that the change of basis matrix $\tabbf{Q}$ is given by
$$
\tabbf{Q}^\prime=
\left(
\begin{array}{cccccc}
\vecbf{a}^{1}, & \ldots, & \vecbf{a}^{d}, & \vecbf{a}^{(d+1)}, & \ldots, & \vecbf{a}^{p}
\end{array}
\right),
$$
thus the left multiplication of $\colbf{y}$ and $r$ by $\tabbf{Q}$, using (\ref{eq:id}), will
give (\ref{eq:can}).
\end{proof}

From this last proposition we can see that the Semi-Linear Auto-Associative models have a relatively
simple geometrical structure and that we cannot expect to model highly non-linear models with them.

\subsection{Probabilistic Semi-Linear Auto-Associative Models (PSLAAM)}
\label{subsec:PSLAAM}

In the sequel, we will denote by $V$ the subspace spanned by the set of vectors
$(\vecbf{a}^1, \ldots, \vecbf{a}^d)$, and give us an arbitrary orthonormal
basis of $V^\bot$ denoted by $(\vecbf{a}^{d+1}, \ldots, \vecbf{a}^p)$. We will denote by
$\tabbf{P}$ the matrix $(\vecbf{a}^{1},\ldots,\vecbf{a}^{d})^\prime$ 
and by $\bar{\tabbf{P}}$ the matrix $(\vecbf{a}^{d+1},\ldots,\vecbf{a}^{p})^\prime$.
As in proposition \ref{prop:SLAAM}, $\tabbf{Q}$ represents the unitary matrix
$(\tabbf{P} | \bar{\tabbf{P}})^\prime = (\vecbf{a}^{1},\ldots,\vecbf{a}^{p})^\prime$.

\subsubsection{General Gaussian Setting}

\begin{defi}\label{def:GSLAAM}
Let $\colbf{x}$ be a $d$-dimensional Gaussian random vector:
\begin{equation}\label{eq:lawx}
\colbf{x} \sim \mathcal{N}(\mu_x, \Sigma_x)
\end{equation}
and let $\tilde{\varepsilon}$ be a $p$-dimensional centered Gaussian random vector with
a diagonal covariance matrix $\Sigma_{\tilde{\varepsilon}} = \mathrm{Diag}(\sigma_1,\ldots,\sigma_p)$.

The $p$-dimensional vector $\colbf{y}$ is a Probabilistic Semi-Linear Auto-Associative Model (PSLAAM)
if it can be written
\begin{equation} \label{eq:model}
\colbf{y} = \tabbf{Q}^\prime
\left(
\begin{pmatrix}
x_1 \\ \vdots \\
x_d \\
\tilde{r}_{d+1}(\colbf{x}) \\ \vdots \\
\tilde{r}_p(\colbf{x})
\end{pmatrix}
+ \tilde{\varepsilon}
\right)
= R(\colbf{x})+ \varepsilon,
\end{equation}
where the $\tilde{r}_j(\colbf{x})$, $d+1 \leq j \leq p$, are arbitrary real functions from $\Rd$ to
$\R$.
\end{defi}

\subsubsection{Link with the Principal Component Analysis}
\label{subsub:PCA}

Assume that:
\begin{enumerate}
\item $\tilde{r}_j(\colbf{x}) = \tilde{\mu}_j$ for all $j\in\{d+1,\ldots p\}$,
\item the covariance matrix of $\colbf{x}$, $\Sigma_x = \mathrm{Diag}(\sigma_1^2, \ldots, \sigma_d^2)$
  is diagonal with $\sigma_1 \geq \sigma_2 \geq \ldots \geq \sigma_d$,
\item the Gaussian noise $\tilde{\varepsilon}$ have the following covariance matrix $\Sigma_\varepsilon
= \mathrm{Diag}(0,\ldots,0,\sigma^2,\ldots,\sigma^2)$ with $\sigma<\sigma_d$.
\end{enumerate}
then the vector $\colbf{y}$ is a Gaussian random vector
$$
\colbf{y} \sim \mathcal{N}(\mu, \Sigma)
$$
with
$$
\mu = \tabbf{Q}'
\begin{pmatrix}
\tilde{\mu}_1 \\ \vdots \\ \tilde{\mu}_d\\ \tilde{\mu}_{d+1} \\ \vdots \\ \tilde{\mu}_p
\end{pmatrix}
\qquad \mbox{and} \qquad
\Sigma = \tabbf{Q}
\begin{pmatrix}
\sigma_1 \\
         & \ddots &          &         & {0}    \\
         &        & \sigma_d \\
         &        &          &  \sigma \\
         & {0}    &          &         & \ddots \\
         &        &          &         &        & \sigma
\end{pmatrix}
\tabbf{Q}'
$$
and $\vecbf{a}^1, \ldots, \vecbf{a}^d$ are the $d$ first eigenvectors given by the PCA.

\subsubsection{Link with the Probabilistic Principal Component Analysis}
The probabilistic PCA \cite{TB99} is a model of the form
\begin{equation}\label{eq:PPCA}
\colbf{y} = \mu + \tabbf{W} \colbf{x}+ \varepsilon,
\end{equation}
with $\tabbf{W}$ a $(p,d)$ matrix, $\colbf{x}$ a $d$-dimensional isotropic Gaussian vector, i.e.
$\colbf{x}\sim \mathcal{N}(0,I_d)$, and $\varepsilon$ a $p$-dimensional centered Gaussian
random vector with covariance matrix $\sigma^2 I_p$. The law of $\colbf{y}$ is not modified if
$\tabbf{W}$ is right multiplied by a $(d,d)$ unitary matrix, it is thus possible to impose to the
rows of $\tabbf{W}$ to be orthogonal (assuming that $\tabbf{W}$ is of full rank).

The following proposition is then straightforward
\begin{prop}
Assume that $\tilde{\varepsilon}$ (and thus $\varepsilon$) is an isotropic Gaussian noise, i.e.
$\Sigma_{\tilde{\varepsilon}} = \sigma^2 I_p$, take $\tilde{r}_j = \tilde{\mu}_{j}$ for all $d+1\leq j\leq p$
and set
$$ \tabbf{W} = \tabbf{P}^\prime
\begin{pmatrix}
\sigma_1  & 0      & \ldots & 0 \\
   0      & \ddots & \ddots & \vdots\\
 \vdots   & \ddots & \ddots & 0 \\
 0        & \ldots &    0   &\sigma_d
\end{pmatrix}
.
$$
The resulting Probabilistic Semi-Linear Auto-Associative Model is a Probabilistic Principal
Component Analysis.
\end{prop}

For this simple model there exists a close form of the posterior probability of
$\colbf{y}$ and for the maximum likelihood of the parameters of the model.
In particular, the matrix $\tabbf{W}$ can be estimated up to a rotation and spans the principal
subset of the data.

%

\section{Semi-Linear PCA}
\label{sec:NLPCA}

Our aim is now to generalize the PCA model we present in part (\ref{subsub:PCA}) to the semi-linear case.
We observe that in the PCA model, if the matrix $\colbf{P}$ is known, then we are able to know the random
variable $\colbf{x}$. This observation lead us to formulate the following hypothesis about the noise
$\tilde{\varepsilon}$:

\begin{description}
\item[N] the Gaussian noise $\tilde{\varepsilon}$ have the following covariance matrix
$\Sigma_{\tilde{\varepsilon}} = \mathrm{Diag}(0,\ldots,0,\sigma^2,\ldots,\sigma^2)$.
\end{description}

Expressing $\colbf{y}$ in the basis  $\mathcal{B}$ (definition \ref{def:GSLAAM}) we get the following
expression for $\tilde{\colbf{y}}$:
\begin{equation} \label{eq:canmodel}
\begin{pmatrix}
\tilde{y}_1     \\ \vdots \\ \tilde{y}_d \\
\tilde{y}_{d+1} \\ \vdots \\
 \tilde{y}_p
\end{pmatrix}
 =
\begin{pmatrix}
x_1 \\ \vdots \\ x_d \\
\tilde{r}_{d+1}(\colbf{x}) \\ \vdots \\
 \tilde{r}_p(\colbf{x})
\end{pmatrix}
+
\begin{pmatrix}
0     \\ \vdots \\ 0 \\
\tilde{\varepsilon}_{d+1} \\ \vdots \\
\tilde{\varepsilon}_p
\end{pmatrix}.
\end{equation}
In other word, the coordinates of $\tilde{\colbf{y}}$ can be split in two sets. The $d$ first
coordinates are the Gaussian random vector $\colbf{x}$, while the $p-d$ remaining coordinates
are a random vector $\colbf{z}$ which is conditionally to $\colbf{x}$ a Gaussian random vector
$\normal{\tilde{r}(\colbf{x} )}{\sigma^2 I_{p-d}}$. Observe that the regression functions
are dependents of the choice of the vectors $\vecbf{a}_{d+1},\ldots,\vecbf{a}_p$ and that, as
the noise $\varepsilon$ lives in the orthogonal of $V$, we have $\colbf{x} = \mathbf{P} \colbf{y}$.

\subsection{Maximum Likehood Estimates}

The parameters we have to estimate are the position and correlation parameters $\mu_x$ and $\Sigma_x$
for the $\colbf{x}$ part and $(\sigma^2, \tilde{r})$ for the non-linear part. Given a set of $n$ points
$\tabbf{Y} = (\colbf{y}_1,\ldots,\colbf{y}_n)'$ in $\Rp$, we get by projection two sets of $n$ points
$\tabbf{X} = (\colbf{x}_1,\ldots,\colbf{x}_n)' = \tabbf{Y} \tabbf{P}'$ in $\Rd$, and
$\tabbf{Z} = (\colbf{z}_1,\ldots,\colbf{z}_n)' = \tabbf{Y} \tabbf{\bar{P}}'$ in $\mathbb{R}^{p-d}$.

Standard calculation give the maximum likehood for $\mu_x$ and $\Sigma_x$
\begin{equation}\label{eq:mu}
\hat{\mu}_x = \frac{1}{n} \sum_{i=1}^n \colbf{x}_i.
\end{equation}
and
\begin{equation}\label{eq:sigma_x}
\hat{\Sigma}_x = \frac{1}{n}
 \sum_{i=1}^n (\colbf{x}_i-\hat{\mu}_x)(\colbf{x}_i-\hat{\mu}_x)'.
\end{equation}

The maximum likehood of $\sigma^2$ is given by
\begin{equation}\label{eq:sigma}
\hat{\sigma}^2 = \frac{1}{n(p-d)}
 \sum_{i=1}^n \left\| \colbf{z}_i - \hat{\tilde{R}}(\colbf{x}_i) \right\|^2.
\end{equation}

It remains to estimate $\tilde{R}$. We consider two cases : $\tilde{R}$ is a linear function
and $\tilde{R}$ is a linear combination of the elements of a B-Spline function basis.
The linear case is just a toy example that we will use for comparison with the
additive B-Spline case. In the non-linear case, we have to estimate a function from
$\Rd$ to $\R^{p-d}$. As we say in the introduction this is a difficult task and we restrict
ourself to a generalized additive model (GAM) by assuming that the function $\tilde{R}$ is
additive, i.e.
\begin{equation}
\label{eq:additive}
\tilde{R}(\colbf{x}) = \sum_{j=1}^d \tilde{\colbf{r}}^j(x_j),
\end{equation}
where each $\tilde{\colbf{r}}^j$ is a map from $\R$ into $\R^{p-d}$.

\subsection{Linear Auto-Associative Models}
In the linear case, we are looking for a vector $\mu$ and a $(d,p-d)$ matrix $\colbf{R}$
minimizing
$$
\sum_{i=1}^n \left\| \colbf{{z}}_i - \mu - \colbf{R}'\colbf{x}_i
\right\|^2.
$$
It is easily verified that
$$
\hat{\mu} = \frac{1}{n}\sum_{i=1}^n \left(\colbf{{z}}_i - \tabbf{R}'\colbf{x}_i\right)
           = \hat{\mu}_z - \tabbf{R}'\hat{\mu}_x.
$$
Setting $\colbf{\bar{X}} = \colbf{X} - \tabbf{1}\hat{\mu}_x^\prime$ and $\colbf{\bar{Z}} = \colbf{X}
- \tabbf{1}\hat{\mu}_z^\prime$, where $\tabbf{1}$ represent a vector of
size $n$ with $1$ on every coordinates.
Assuming that the matrix $\colbf{\bar{X}}^\prime\colbf{\bar{X}}$
is invertible, standard calculus show that
$$
 \hat{\colbf{R}} = (\colbf{\bar{X}}^\prime\colbf{\bar{X}})^{-1}
\colbf{\bar{X}}^\prime\colbf{\bar{Z}}.
$$

Finally, using the decomposition in eigenvalues of the covariance matrix of $\tabbf{Y}$, it is
straightforward to verify the following theorem
\begin{theo}
\label{th:pca}
If the $d$ orthonormal
vectors $\colbf{a}^{1},\ldots, \colbf{a}^{d}$ are the eigenvectors associated
with the $d$ first eigenvalues of the covariance matrix of $\colbf{{Y}}$
then the estimated auto-associative model is the on obtain by the PCA.
\end{theo}

\subsection{Additive Semi-Linear Auto-Associative Models}
\label{subsec:ASLAAM}
In order to estimate the regression functions $\colbf{\tilde{r}})^j, j=1\ldots,d$, we express
them as a linear combination of $m$ B-Spline functions basis $s^{jl}$ where $m$ is a number
chosen by the user.  We have thus to estimate the set of coefficients
$(\alpha_{jl})$, $j=1,\ldots,d$, $l=0,\ldots m$ by minimizing
$$
\sum_{i=1}^n \left\| \colbf{z}_i - \alpha_0 - \sum_{j=1}^d  \sum_{l=1}^m \alpha_{jl} \colbf{s}^{jl}(x_{ij})
\right\|^2.
$$
Standard regression techniques give then the estimates
$$
\hat{R}(\colbf{x}) = \hat{\alpha}_0 + \sum_{j=1}^d \sum_{l=1}^m \hat{\alpha}_{jl} \colbf{s}^{jl}(x_{ij}),
\qquad \text{ with } \hat{\alpha} = ((\colbf{{S}}^\prime\colbf{{S}})^{-1}\colbf{{S}}^\prime\colbf{Z})
$$
where $\colbf{S}$ is the design matrix which depends
of the knots position, degree of the B-Spline and the number of control points chosen by the user
\cite{Prautzsch2002}.

The estimated regression function $\colbf{r}^j$, for $j=1,\ldots,d$ are then given by the formula
$$\hat{\colbf{r}}^j = \sum_{l=1}^m \hat{\alpha}_{jl} \colbf{s}^{jl}. $$

\subsection{Estimation in practice}
The drawback of the previous maximum likehood equations is that, given the projection matrix
$\colbf{P}$, we have to perform a rotation of the original data set and next to perform an
inverse rotation of the estimated model. In practice, we avoid such computations by estimating
the model using the following steps:
\begin{itemize}
\item {\bf (C)} Center and (optionally) standardize the data set $\mathbf{Y}$: obtain $\bar{\tabbf{Y}}$,
\item {\bf (P)} Compute the projected data set $\tabbf{X} = \bar{\tabbf{Y}} \tabbf{P}'$ ($\tabbf{X}$ is centered),
\item {\bf (R)} Compute the regression $\bar{\tabbf{Y}} \sim \tabbf{X}$ (without intercept),
\item {\bf (S)} Compute the log-likelihood and the BIC criteria.
\end{itemize}

As we can see the main difference is in the regression step: we estimate directly a function
from $\Rd$ to $\Rp$. In practice, as the non-linear part of the model is in $V^\bot$,
the regression function we obtain numerically give the identity function in the $V$ space.

\subsection{Model Selection}

Since a Semi-linear PCA model depends highly of the projection matrix $\mathbf{P}$, model
selection allows to select among various candidate the best projection.
Several criteria for model selection have been proposed in the literature and the widely
used are penalized likelihood criteria. Classical tools for model selection include the
AIC \cite{Akaike74} and BIC \cite{Schwarz1978} criteria. The Bayesian Information Criterion
(BIC) is certainly the most popular and consists in selecting the model which penalizes the
likelihood by $\frac{\gamma(\mathcal{M})}{2} \log(n)$ where $\gamma(\mathcal{M})$ is the number
of parameters of the model $\mathcal{M}$ and $n$ is the number of observations.

In practice we will fix a set of vectors $\vecbf{a}^1,\ldots, \vecbf{a}^{d_{\max}}$ given either
by the contiguity analysis (section \ref{sec:contiguityAnalysis}) or the PCA and select the dimension
of the model using the BIC criteria because of its popularity. The projection matrices we compare are
thus $\mathbf{P}=(\vecbf{a}^1)'$, $\mathbf{P}=(\vecbf{a}^1,\vecbf{a}^2)'$,... and so on.

\section{Contiguity Analysis}
\label{sec:contiguityAnalysis}

%
Given $(\vecbf{a}^1, \ldots, \vecbf{a}^{d})$ an orthonormal set of vector in $\Rp$,
an index $I$:~$\R^{p\times d}\to\R_+$ is a functional measuring the interest
of the projection of the vector $\vecbf{y}$ on $\mathrm{Vec}({\vecbf{a}^1},\ldots,{\vecbf{a}^d})$
with a non negative real number. A widely used choice of $I$ is
$I(\proscal{\vecbf{a}^1}{\colbf{y}},\ldots,\proscal{\vecbf{a}^d}{\colbf{y}})=
tr(\V{\colbf{P}\colbf{y}})$, the projected variance. This is the criteria maximized in the usual
PCA method \cite{Jolli}.

The choice of the index $I$ is crucial in order to find
"good" parametrization directions for the manifold to be estimated.
We refer to \cite{Huber} and \cite{Jones} for a review on this topic
in a projection pursuit setting. The meaning of the word "good" depends
on the considered data analysis problem. For instance, Friedman {\it et al}
\cite{Fried74, Fried87}, and more recently Hall \cite{Hall}, have proposed an
index which measure the deviation from the normality in order to
reveal more complex structures of the scatter plot.
An alternative approach can be found in \cite{Caussinus} where a
particular metric is introduced in PCA in order to detect clusters.
We can also mention the index dedicated to outliers detection \cite{Pan}.

Our approach generalizes the one we present in \cite{GI} and consists in
defining a contiguity coefficient similar to Labart one's \cite{Lebart00} whose
maximization allows to unfold nonlinear structures.

A \emph{contiguity} matrix is a $n\times n$ boolean matrix $M$ whose
entry is $m_{ij}=1$ if data points $i$ and $j$ are "neighbors" and $m_{ij}=0$
otherwise. Lebart proposes to use a threshold $r_0$ to the set of $n(n-1)$
distances in order to construct this matrix but the choice of $r_0$
could be delicate. In \cite{GI} we propose to use a first order contiguity matrix,
i.e. $m_{ij} = 1$ iff $i$ is the nearest neighbor of $j$ in order to construct the proximity
graph. In order to get a more robust estimate of the neighbor structure, it is possible
to generalize this approach and to use a $k$-contiguity matrix, i.e.
$m_{ij} = 1$ iff $i$ is one of the k-nearest neighbor of $j$.

The contiguity matrix being chosen, we compute the \emph{local covariance} matrix
\begin{equation}
\tabbf{V}^* = \frac{1}{2kn} \sum_{i=1}^n \sum_{j=1}^n m_{ij}
(\colbf{y}_i - \colbf{y}_j)(\colbf{y}_i - \colbf{y}_j)'
\end{equation}
and the total variance matrix
\begin{equation}
\tabbf{V} = \frac{1}{n} \sum_{i=1}^n
(\colbf{y}_i - \hat{\mu})(\colbf{y}_i - \hat{\mu})'
\end{equation}
The axis of projection are then estimated by maximizing the contiguity index
\begin{equation}
I(\vecbf{a}^1, \ldots, \vecbf{a}^{d}) = \sum_{i=1}^{d}
\frac{\vecbf{a}^{i\prime}\tabbf{V}^* \vecbf{a^{i}}}{\vecbf{a}^{i\prime}\tabbf{V} \vecbf{a}^{i}}.
\end{equation}
Using standard optimization techniques, In can be shown that the resulting axis are the $d$
eigenvectors associated with the largest eigenvalues of the matrix $V^{*-1}V$.

\section{Examples}
\label{sec::apps}

We first present two illustrations of the estimation principle of PSLAAM on low dimensional data
(Section~\ref{sub:parun} and \ref{sub:pardeux}). These two simulated examples are very similar from the
one we use in our previous article with S. Girard \cite{GI}. Second, PSLAAM is applied to an
astronomical analysis problem in Section~\ref{sub:partrois}.

Similarly, we always use an additive B-Spline regression model for the estimation
of the regression function $\tilde{R}$ (section \ref{subsec:ASLAAM}). The B-Spline are of degree
3 and we select the number of control points using the BIC .

\subsection{First example on simulated data}
\label{sub:parun}

The data are simulated using a one-dimensional regression function in $\R^3$.
The equation of the AA model is given by
\begin{equation}
\label{tirebouch}
x \to (x, \sin{x}, \cos{x}),
\end{equation}
and thus $P(x,y,z)=x$ .
The first coordinate of the random vector is sampled from a centered Gaussian distribution
with standard deviation $\sigma_x = 3$ a thousand times. An independent noise with standard
deviation $\sigma = 1$ has then been added to the $y$ and $z$ coordinates.

The axis of projection have been computed thanks to the contiguity analysis (section
\ref{sec:contiguityAnalysis}) using the 3 nearest neighbors for the proximity graph.
The correlations between the projected data set and the original data set are
\begin{center}
\begin{tabular}{lccc}
         & X            & Y            & Z\\
Proj$_1$ & 0.9999680850 & -0.0005794581 & 0.0089238830
\end{tabular}
\end{center}
which show that the first axis given by the contiguity analysis is very close from the $x$-axis
as it was expected. The result of the contiguity analysis can be visualized in the figure \ref{fig:tirebouchon1}.
\begin{figure}[htb]
\includegraphics[width =11cm]{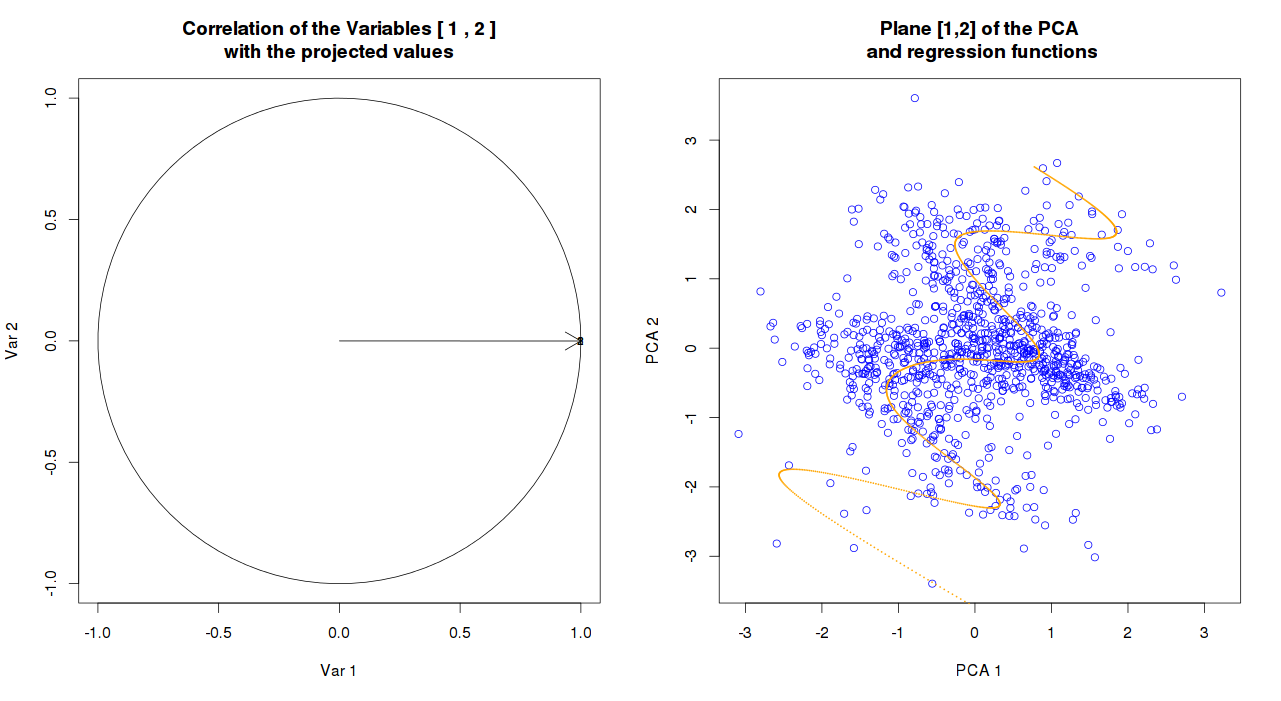}
\caption{Correlation of the first axis obtain using a contiguity analysis with the $X$ and $Y$ variables
and representation of the scatter-plot and regression function in the main PCA plan. This graphic has been obtained
with R using the \emph{plot} command of the aam library.}
\label{fig:tirebouchon1}
\end{figure}

The projected variance on the first axis is $9.20496$ which is also very close from $9$.
We use the BIC criteria in order to select the dimension of the model and the number of control points.
A summary of the tested model is given in the table \ref{tab:simul1d}
\begin{center}
\begin{table}[htb]
\definecolor{tcA}{rgb}{0,1,1}
\begin{tabular}{|l|c|r|r||r|r|}
\hline
   &     & \multicolumn{2}{c||}{Contiguity Analysis} & \multicolumn{2}{c}{PCA} \\
\hline
      & dim & BIC        & Residual Variance & BIC     & Residual Variance\\
\hline
linear & 1  & 9987,13(5) & 1.439             & 11495.8 & 1.43864 \\
\hline
       & 2  & 10609.7(10)& 1.40815  & & \\
\hline
\hline

9 & 1 & 11247.6(29) & 1.06557  & 11058.4 & 1.06407 \\
\hline
  & 2 & 11621(58) & 1.1086 & &\\
\hline

10 & 1 & 11060.7(32) & 0.986316 & 10926.3 & 0.985777 \\
\hline
   & 2 & 11469(64) & 0.913605 & &\\
\hline

11 & 1 & 10853.7(35) & 0.941064 & 10855 & 0.92711 \\
\hline
   & 2 & 11503(70) & 0.90682 & & \\
\hline

 12 & 1 & \multicolumn{1}{>{\columncolor{tcA}}r } {10844.6(38)} & 0.92717
        & \multicolumn{1}{>{\columncolor{tcA}}r } {10845(38)} & 0.941661 \\
\hline
  & 2  & 11525(76) & 0.892669 & &\\
\hline

13 & 1 & 10871.4(41) & 0.929965 & 10871.6 & 0.927976 \\
\hline
   & 2  & 11568.4(82) & 0.891119 & & \\
\hline

14 & 1 & 108887.5(44) & 0.927834 & 10888.3 & 0.927976 \\
\hline
   & 2  & 11604(88) & 0.885939  & &\\
\hline
\end{tabular}
\caption{Values of the BIC criteria for $d=1$ and $d=2$ and for various number of control points (given in the
first column). The number of free parameters of each model is given in parenthesis. The BIC criteria selects the
model of dimension 1 with 11 control points. The axis of projection can be either the one obtained by contiguity analysis
or the one obtained using the PCA.}
\label{tab:simul1d}
\end{table}
\end{center}

Finally the result of the regression is drawn in the figure \ref{fig:tirebouchon2}.
\begin{figure}[h]
\includegraphics[height = 4cm]{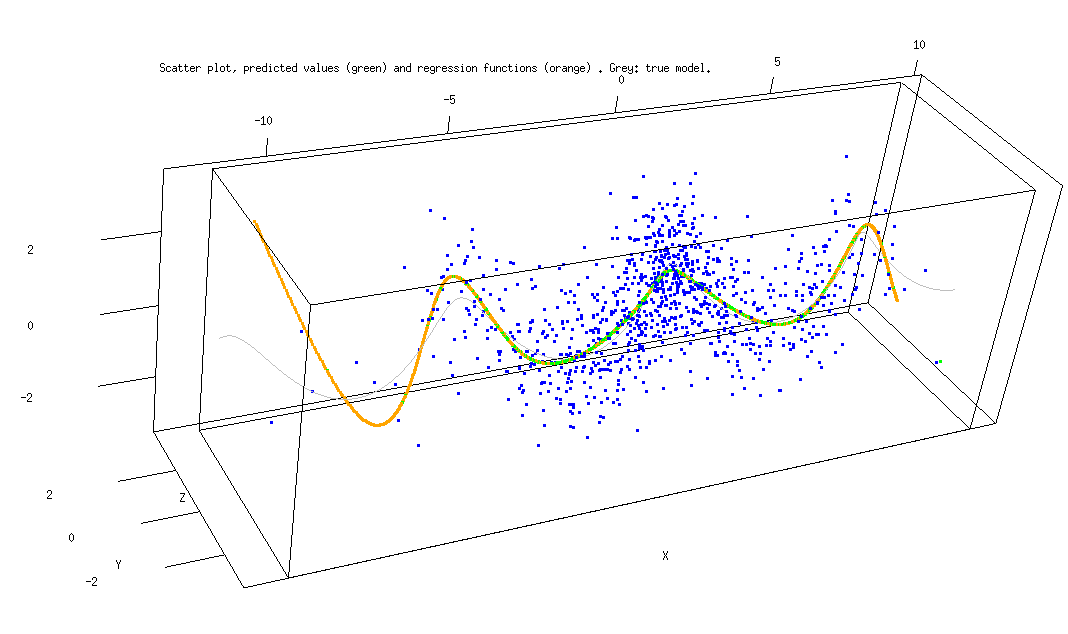}
\caption{The simulated scatter-plot (blue), the estimated AAM (orange) and the true AAM (grey).
This graphic is obtained with R using the \emph{draw3d} command of the aam library.}
\label{fig:tirebouchon2}
\end{figure}

\subsection{Second example on simulated data}
\label{sub:pardeux}
In our second example the AAM is given by
\begin{equation}
\label{chap}
(x,y)  \to \left(x,y,\cos(\pi r/3) (1 - \exp(-64r^2)) \exp(0.2r) \right)
\end{equation}
with $r=\sqrt{x^2+y^2}$ and thus $P(x,y,z)=(x,y)$. The first two coordinates of the random vector are
sampled from a centered Gaussian distribution with covariance matrix
$$ \Sigma_x = \begin{pmatrix}
                1.8 & 0 \\
                0   & 1.5 \\
              \end{pmatrix}
$$
and $n=1000$ points are simulated. An independent noise with standard
deviation $\sigma = 0.5$ has then been added to the $z$ coordinate.

The correlations between the projected data set and the original data are
\begin{center}
\begin{tabular}{llll}
         & X            & Y            & Z\\
Proj$_1$ & 0.99924737   & -0.1488330 & 0.0179811 \\
Proj$_2$ & -0.14239437 & 0.98532966  & 0.01396622
\end{tabular}
\end{center}
which show that the $(x,y)$-plan is the plan essentially chosen by the contiguity analysis.
The result of the contiguity analysis is displayed in the figure \ref{fig:hat}.

\begin{figure}[htb]
\includegraphics[width = 12cm]{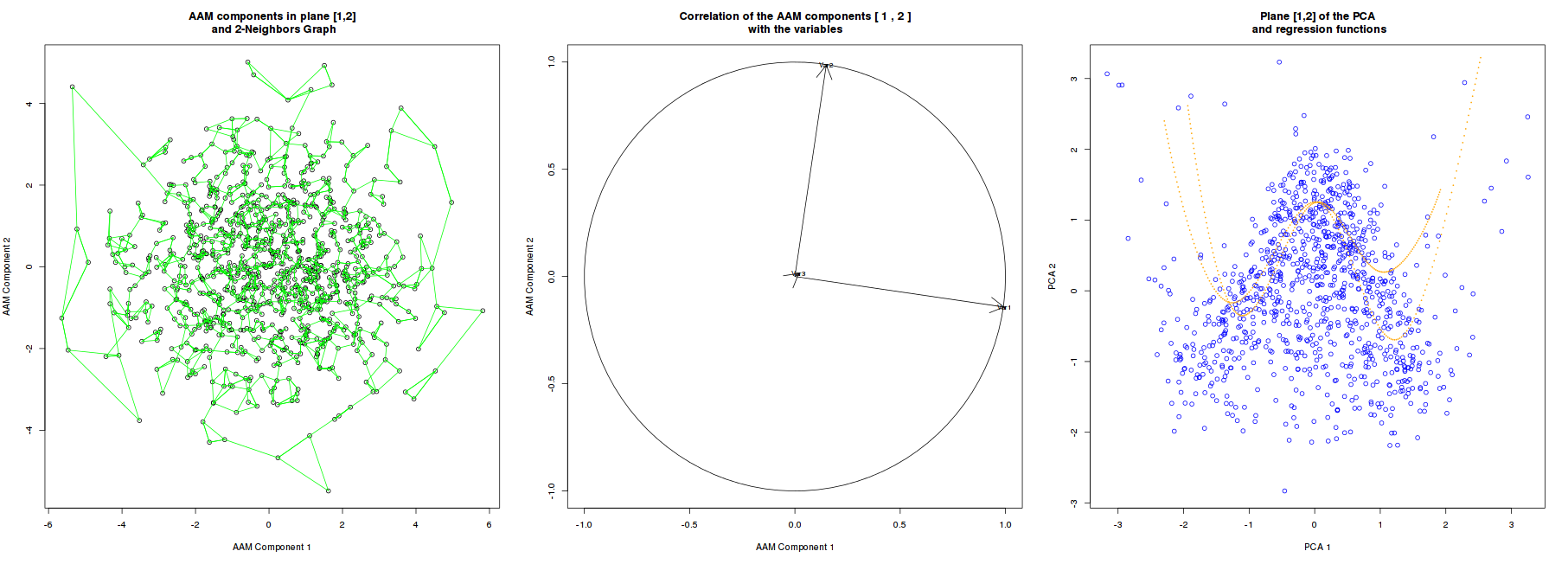}
\caption{The AAM components and the 2-neighbors graph, the correlation circle of the AAM components with the variables
and representation of the scatter-plot in the PCA plan. These pictures have been obtained using the \emph{plot} command
of the aam library.}
\label{fig:hat}
\end{figure}

The BIC criteria select 7 control points. A summary of the tested model is given in the table \ref{tab:simul2d}.
The selected model over-estimate the residual variance by a factor 2. It is not a surprising result as the original
model is not additive and we cannot expect to reconstruct it exactly. We don't show the results with the PCA
as the first axis this method select is the $Z$-axis which is clearly the wrong parametrization.

\begin{center}
\definecolor{tcA}{rgb}{0,1,1}
\begin{table}[htb]
\begin{tabular}{|l|c|r|r|}
\hline
 & {dim} & {BIC} & {Residual Variance}\\
\hline
{linear} & {1} & 11024.7(5) & 1.78335\\
\hline
 & {2} &  10.832.7(10) & 1.20314\\
\hline
\hline

 & {2} & {10654.2(34)} & {0.852753}\\
\hline
{6} & {1} & 11072.6(20) & {1.68675}\\
\hline
 & {2} & 10716.1(40) & {0.870381}\\
\hline

{7} & {1} & 11014.2(23) & {1.55841}\\
\hline
\rowcolor{tcA}
   & {2} & {10.213.6(46)} & {0.505186}\\
\hline

{8} & {1} & 11.031.6(26) & {1.55334}\\
\hline
 & {2} & {10.216.6(52)} & {0.486137}\\
\hline
{9} & {1} & 11051.7(29) & {1.55229}\\
\hline
 & {2} & {10255(58)} & {0.484644}\\
\hline
\end{tabular}
\caption{Values of the BIC criteria for $d=1$ and $d=2$ and for various number of control points (given in the
first column). The number of free parameters of each model is given in parenthesis. The BIC criteria selects the
model of dimension 2 with 7 control points.}
\label{tab:simul2d}
\end{table}
\end{center}

The true model and the estimated model obtained with an additive B-Spline regression
are given in the figure \ref{fig:hat3}.
\begin{figure}[htb]
\begin{minipage}{0.32\textwidth}
\includegraphics[width = 3.7cm]{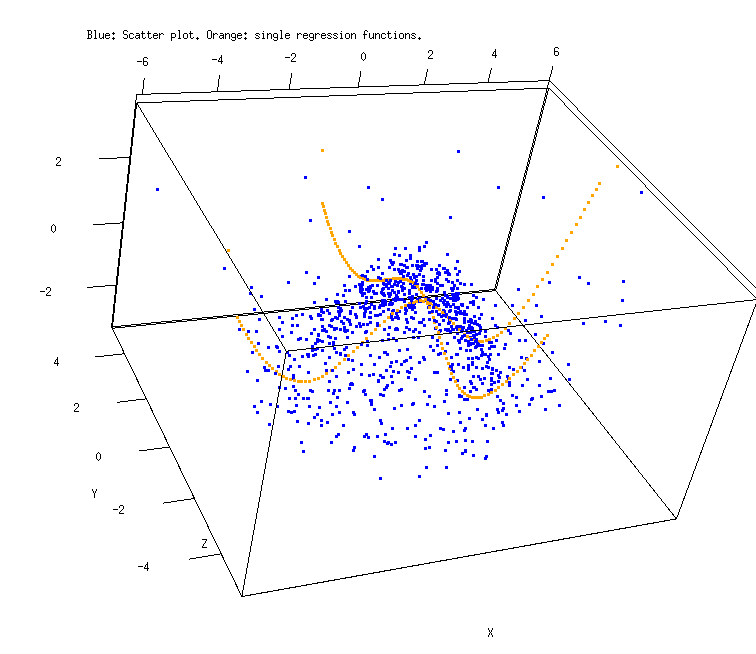}
\begin{center}
(a)\end{center}
\end{minipage}
\begin{minipage}{0.32\textwidth}
\includegraphics[width = 3.7cm]{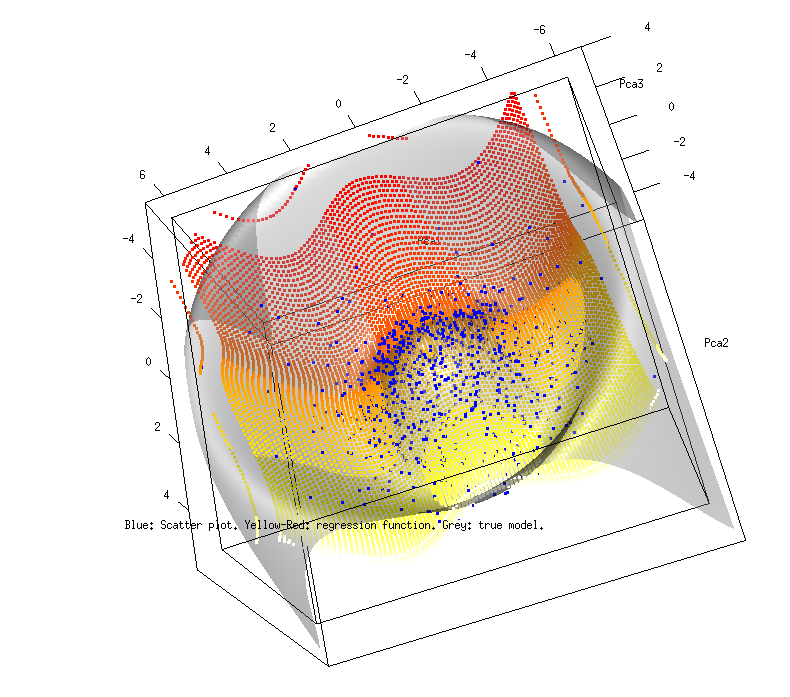}
\begin{center}
(b)\end{center}
\end{minipage}
\begin{minipage}{0.32\textwidth}
\includegraphics[width = 3.7cm]{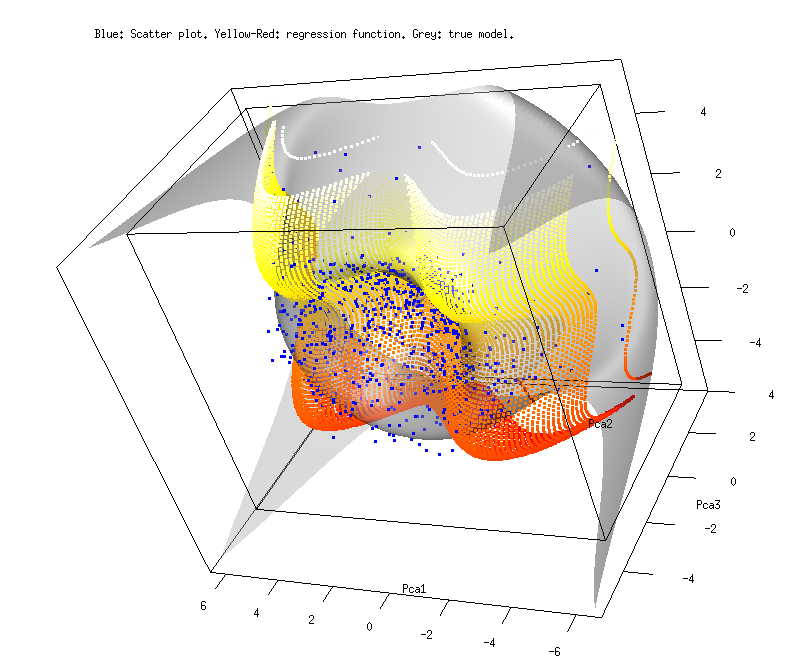}
\begin{center}
(c)\end{center}
\end{minipage}
\label{fig:hat3}
\caption{(a) the original data set (blue) and the individual regression functions. (b)-(c) two views
of the original manifold (grey) and of the extrapolated manifold using an additive B-Spline regression
(yellow-red). These images have been obtained using the \emph{draw3d} function of the aam library.}
\end{figure}

\subsection{Example in spectrometry analysis}
\label{sub:partrois}

Finally we illustrate the performance of the semi-linear PCA on a real data set.
The data consists of $19$-dimensional spectral information of 487 stars
\cite{Stock1999, Stock2004, Stock2008, Garcia2008}
and they have been classified in 6 groups. They have been modeled by \cite{Scholz2007} using an
auto-associative neural networks based on a 19-30-10-30-19 network. Using the terminology of this article
the model proposed by M. Scholz and its co-authors is an auto-associative model of dimension 10.
We select the model using the BIC criteria. The main results are the following:
\begin{enumerate}
\item The axis of projection given by the PCA outperform largely the results we obtain with the
contiguity analysis, for any choice of control point.
\item The BIC criteria retains a model of dimension 5 with 9 Control Points (871 parameters)
when we use a non-linear regression step. The residual variance is $\sigma^2 =0.0080763$
while the total variance (inertia) of the data was 26.59832.
\item The BIC criteria retains a model of dimension 12 (307 parameters)
when we use a linear regression step. Observe that in this case, we are performing
an usual PCA (theorem \ref{th:pca}).
\end{enumerate}

The data cloud in the main PCA space with the values predicted by the model is displayed in the figure
\ref{fig:spect_predicted}.
\begin{figure}[!htb]
\includegraphics[width =9cm]{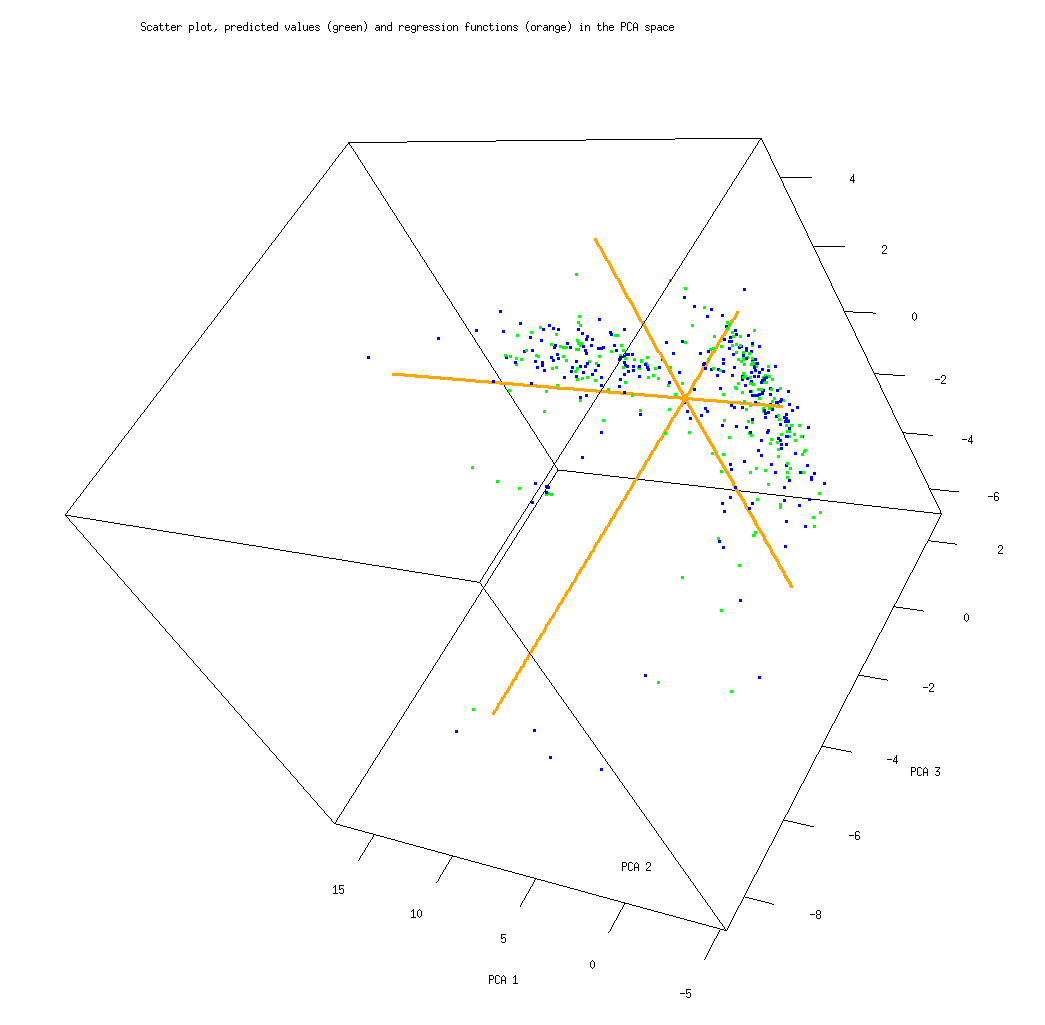}
\caption{Data cloud in the main PCA space with the predicted values (green).
}
\label{fig:spect_predicted}
\end{figure}

A summary of some tested model are given in the table \ref{tab:spectrometry} and
some planes of projection are displayed in the picture \ref{fig:spect_summary}.
We visualize each components of the regression functions by setting all, except one,
predictors to zero and we represent the evolution of the regression function
in the $\R^{14}$ space in the graphic \ref{fig:spect_plot1}.

\begin{center}
\definecolor{tcA}{rgb}{0,1,1}
\begin{table}[!htb]
\begin{tabular}{|l|r|r|r|}
\hline
      \multicolumn{4}{|c|}{PCA} \\
\hline
Control Points & BIC value  &   dim   & Residual variance \\
\hline
Linear       & 1829,95(307) & 12     & 0,0049702 \\
\hline
    7        & 1187,26(820) & 6   &   0,00727521 \\
\hline
    8        & 1147,62(776) & 5    &  0,00975073 \\
\hline
    9        & 453,387(871) & 5     & 0,0080763 \\
\hline
    10       & 701,769(966) & 5     &  0,00768342 \\
\hline
    11       & 1333,45(1061) & 5   &  0,00773327 \\
\hline
\end{tabular}
\caption{Values of the BIC criteria for various number of control points (given in the
first column). The BIC criteria select the model of dimension 5 with 9 control points using as
projection matrix the $5$ axis given by the PCA. The total variance (inertia) of the data set was 26.59832.}
\label{tab:spectrometry}
\end{table}

\begin{figure}[!htb]
\includegraphics[width =11cm]{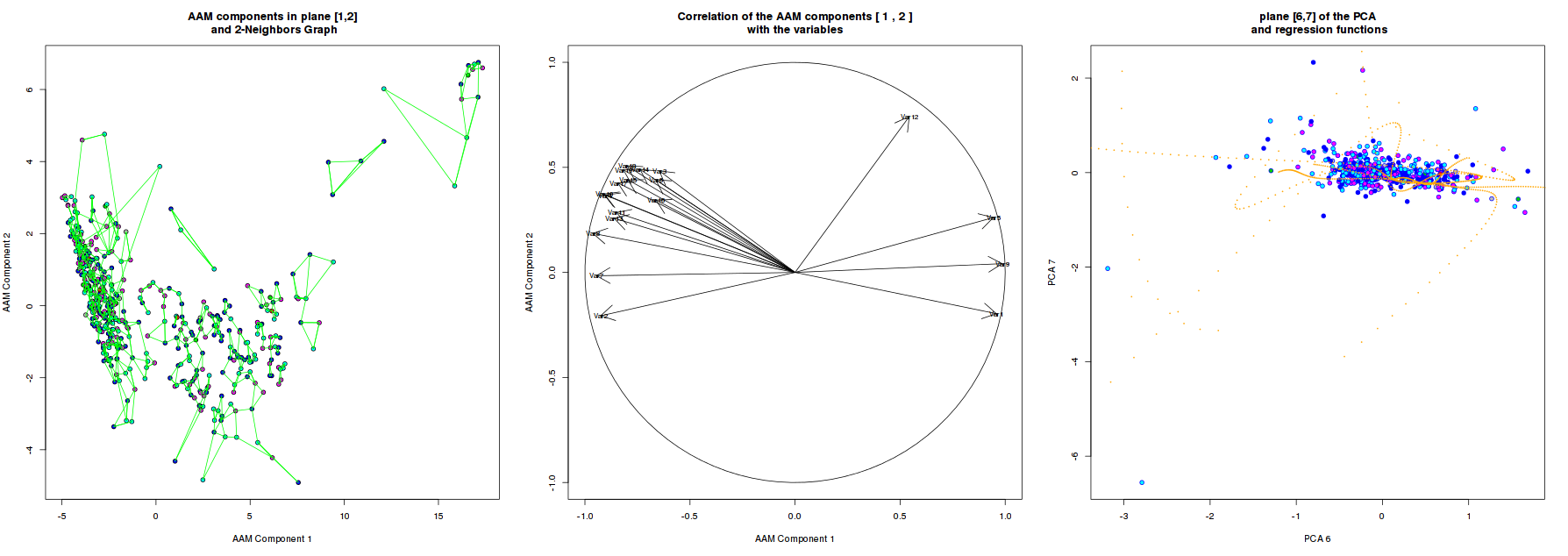}
\includegraphics[width =11cm]{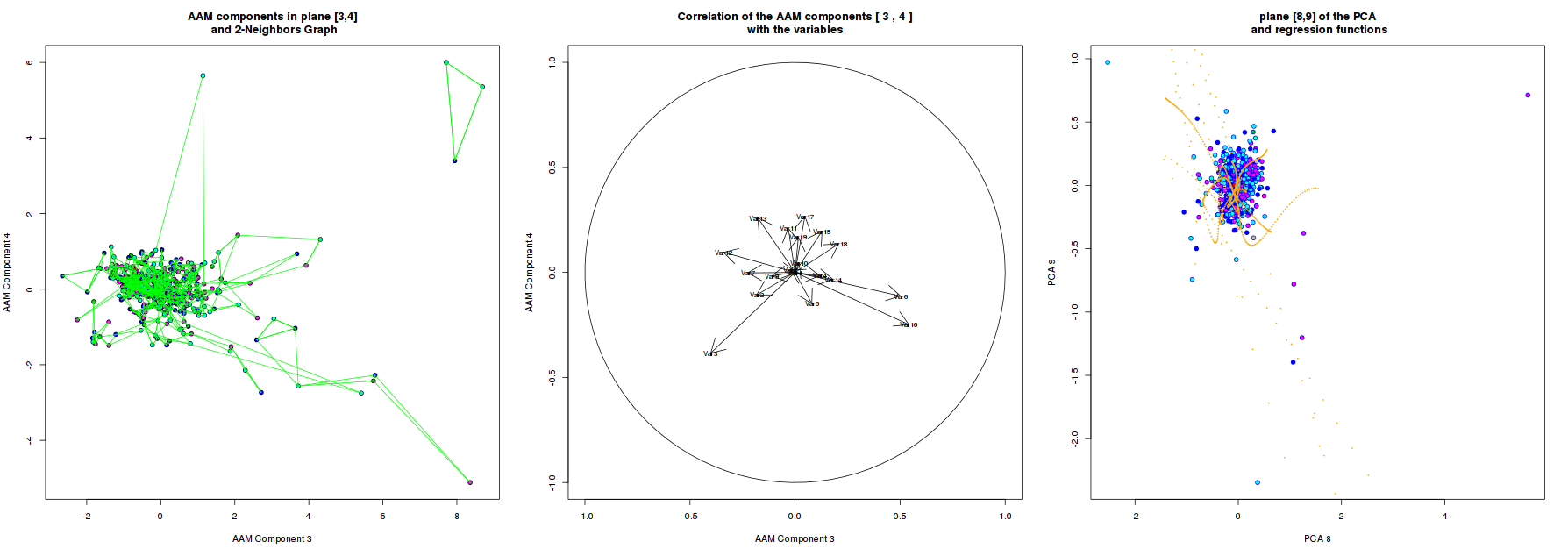}
\caption{Some plans of projection of the PCA with the 2-neighbors graph (left) and with the single
regression functions (right). As we use the PCA for the projection matrix, the AAM components are the
PCA components. The colors of the points represent the classification of the stars.}
\label{fig:spect_summary}
\end{figure}

\begin{figure}[!htb]
\includegraphics[width =11cm]{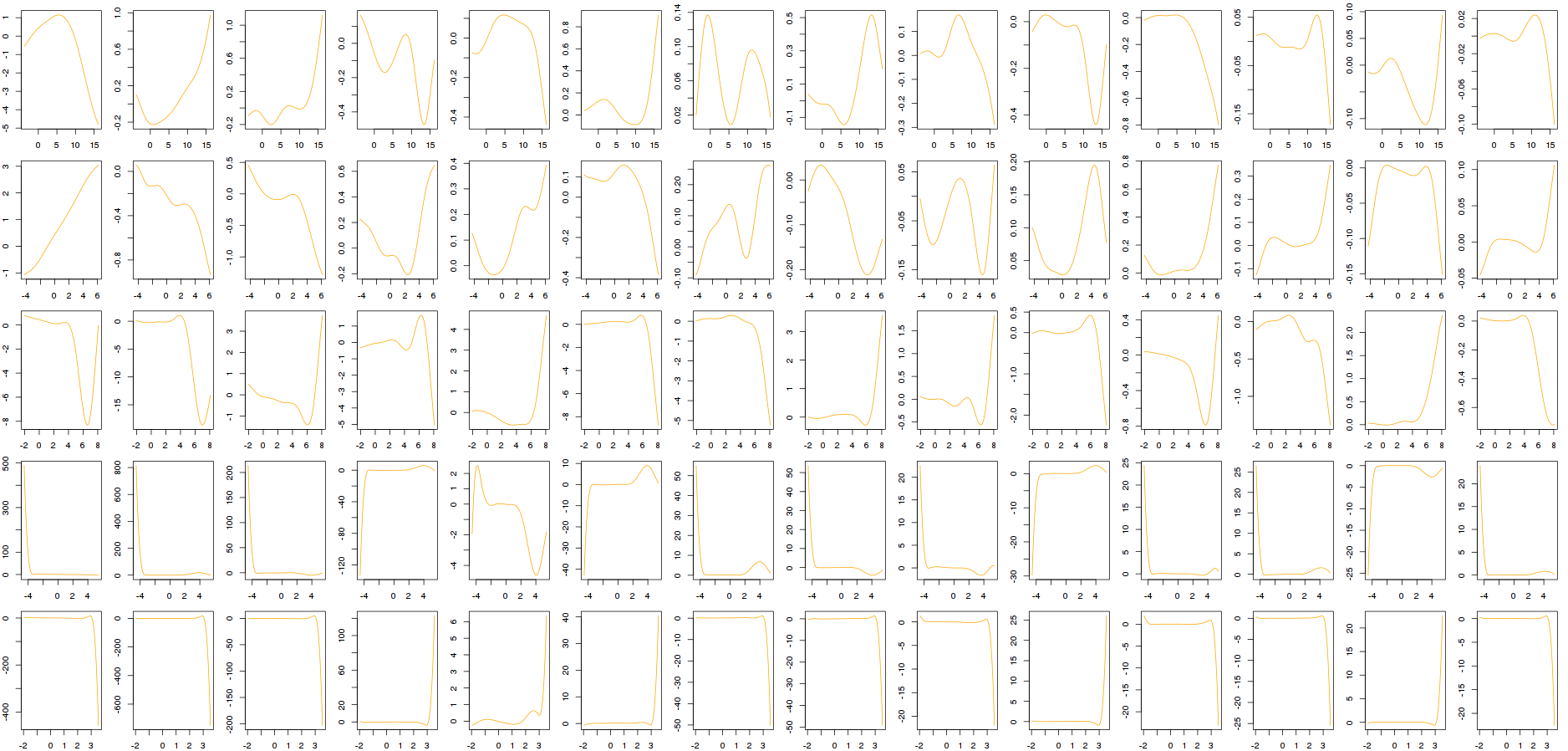}
\caption{The individual regression functions from $\R^5$ to $\R^{14}$. In each row we have the non-zero
predictor sampled in the range $[\min,\max]$, and in each column the evolution of the functions in the
dimension 6,...,19. The system of coordinates is the one given by the PCA. }
\label{fig:spect_plot1}
\end{figure}

\end{center}

\section{Conclusion}
\label{sec::conc}
We have presented a class of auto-associative model for data modeling and visualization
called semi-linear auto-associative models. We provided theoretical groundings for these
models by proving that the principal component analysis and the probabilistic principal component
analysis are special cases. Our model allows to models data set with a simple non-linear component
and is truly generative with an underlying probabilistic interpretation. However it does not
allow to models data with a strong non-linear component and it depends highly on the choice
of the projection matrix.

The Semi-Linear PCA have been implemented in C++ using the \emph{stk++} library \cite{stk++} and
is available at:
\url{https://sourcesup.renater.fr/projects/aam/}.

The program is accompanied with a set of \emph{R} scripts which allows to simulate and display the results
of the \emph{aam} program.


\bibliographystyle{plain}
\bibliography{aam}

\begin{thebibliography}{10}

\bibitem{Akaike74}
H.~Akaike.
\newblock {A new look at the statistical mode identification}.
\newblock {\em IEEE Transaction on Automatic Control}, 19:716--723, 1974.

\bibitem{Baldi}
Pierre Baldi and Kurt Hornik.
\newblock {Neural networks and principal component analysis: Learning from
  examples without local minima}.
\newblock {\em Neural Networks}, 2(1):53--58, 1989.

\bibitem{BES95}
P.~Besse and F.~Ferraty.
\newblock {A fixed effect curvilinear model}.
\newblock {\em Computational Statistics}, 10(4):339--351, 1995.

\bibitem{BISHOP95}
C.~M. Bishop.
\newblock {\em {Pattern recognition and machine learning, information science
  and statistics}}.
\newblock Springer, Berlin, 2006.

\bibitem{Caussinus}
H.~Caussinus and A.~Ruiz-Gazen.
\newblock {Metrics for finding typical structures by means of Principal
  Component Analysis}.
\newblock {\em Data science and its Applications, Harcourt Brace Japan}, pages
  177--192, 1995.

\bibitem{DEL01}
P.~Delicado.
\newblock {Another look at Principal curves and surfaces}.
\newblock {\em Journal of Multivariate Analysis}, 77:84--116, 2001.

\bibitem{DUR93}
J.~F. Durand.
\newblock {Generalized principal component analysis with respect to
  instrumental variables via univariate spline transformations}.
\newblock {\em Computational Statistics and Data Analysis}, 16:423--440, 1993.

\bibitem{Fried87}
J.~H. Friedman.
\newblock {Exploratory Projection Pursuit}.
\newblock {\em Journal of the American Statistical Association},
  82(397):249--266, 1987.

\bibitem{Fried81}
J.~H. Friedman and W.~Stuetzle.
\newblock {Projection Pursuit Regression}.
\newblock {\em Journal of the American Statistical Association},
  76(376):817--823, 1981.

\bibitem{Fried74}
J.~H. Friedman and J.~W. Tukey.
\newblock {A Projection Pursuit algorithm for exploratory data analysis}.
\newblock {\em IEEE Trans. on computers}, 23(9):881--890, 1974.

\bibitem{Garcia2008}
J.~Garcia, N.~Sanchez, and R.~Velasquez.
\newblock {Quantitative Stellar Spectral Classification. IV. Application to the
  Open Cluster IC 2391}.
\newblock {\em Rev.Mex.Astron.Astrofis. 45 (2009) 13-24}, September 2008.

\bibitem{Stock2004}
J.~Garcia, J.~Stock, M.~J. Stock, and N.~Sanchez.
\newblock {Quantitative Stellar Spectral Classification. III. Spectral
  Resolution}.
\newblock {\em Rev.Mex.Astron.Astrofis. 41 (2005) 31-40}, October 2004.

\bibitem{GI}
S.~Girard and S.~Iovleff.
\newblock {Auto-Associative models and generalized principal component
  analysis}.
\newblock {\em Journal of Multivariate Analysis}, 93:21--39, 2005.

\bibitem{Prautzsch2002}
Prautzsch H., Boehm W., and Paluszny M.
\newblock {\em {B{\'e}zier and B-Spline Techniques}}.
\newblock {Mathematics and visualization}. 2002.

\bibitem{Hall}
P.~Hall.
\newblock {On polynomial-based projection indices for exploratory projection
  pursuit}.
\newblock {\em The Annals of Statistics}, 17(2):589--605, 1990.

\bibitem{HAS89}
T.~Hastie and W.~Stuetzle.
\newblock {Principal curves}.
\newblock {\em Journal of the American Statistical Association},
  84(406):502--516, 1989.

\bibitem{Hastie2001}
T.~Hastie, R.~Tibshinari, and J.~Friedman.
\newblock {\em {The elements of statistical learning}}.
\newblock Springer Series in Statistics, Springer, second edition edition,
  2001.

\bibitem{Hastie90}
T.~J. Hastie and R.~J. Tibshirani.
\newblock {Generalized Additive Models}.
\newblock {\em Monographs on Statisics and Applied Probability}, 43, 1990.

\bibitem{Hinton97}
G.~E. Hinton, P.~Dayan, and M.~Revow.
\newblock {Modeling the manifolds of images of handwrittent digits}.
\newblock {\em IEEE transactions on Neural networks}, 8(1):65--74, 1997.

\bibitem{Hotel}
H.~Hotelling.
\newblock {Analysis of a complex of statistical variables into principal
  components}.
\newblock {\em Journal of Educational Psychology}, 24:417--441, 1933.

\bibitem{Huber}
P.~J. Huber.
\newblock {Projection Pursuit}.
\newblock {\em The Annals of Statistics}, 13(2):435--475, 1985.

\bibitem{stk++}
Serge Iovleff.
\newblock {The Statitiscal ToolKit}.
\newblock \url{http://www.stkpp.org/}, 2012.

\bibitem{Jolli}
I.~Jolliffe.
\newblock {\em {Principal Component Analysis}}.
\newblock Springer-Verlag, New York, 1986.

\bibitem{Jones}
M.~C. Jones and R.~Sibson.
\newblock {What is projection pursuit?}
\newblock {\em Journal of the Royal Statistical Society, Ser. A}, 150:1--36,
  1987.

\bibitem{Klinke}
S.~Klinke and J.~Grassmann.
\newblock {Projection pursuit regression}.
\newblock {\em Wiley Series in Probability and Statistics}, pages 471--496,
  2000.

\bibitem{Lebart00}
Lebart L.
\newblock {Contiguity analysis and classification}.
\newblock In Gaul W., Opitz O., and Schader M., editors, {\em {Data Analysis}},
  pages 233--244. Springer-Verlag, 2000.

\bibitem{Bei}
Bei-Wei Lu and Lionel Pandolfo1.
\newblock {Quasi-objective nonlinear principal component analysis}.
\newblock {\em Neural Networks}, 24(2):159--170, 2010.

\bibitem{Pan}
J-X. Pan, W-K. Fung, and K-T. Fang.
\newblock {Multiple outlier detection in multivariate data using projection
  pursuit techniques}.
\newblock {\em Journal of Statistical Planning and Inference}, 83(1):153--167,
  2000.

\bibitem{Pearson}
K.~Pearson.
\newblock {On lines and planes of closest fit to systems of points in space}.
\newblock {\em The London, Edinburgh and Dublin philosophical magazine and
  journal of science}, Sixth Series(2):559--572, 1901.

\bibitem{Roweis2000}
Sam~T. Roweis and Lawrence~K. Saul.
\newblock {Nonlinear Dimensionality Reduction by Locally Linear Embedding}.
\newblock {\em Science}, 290(5500):2323--2326, 2000.

\bibitem{Schoelkopf99}
B.~Sch{\"o}lkopf, A.~Smola, and K.-R. M{\"u}ller.
\newblock {Kernel Principal Component Analysis}.
\newblock In {\em {Advances in Kernel Methods---Support Vector Learning}},
  pages 327--352, 1999.

\bibitem{Scholz2007}
M.~Scholz, M.~Fraunholz, and J.~Selbig.
\newblock {Nonlinear Principal Component Analysis: Neural Network Models and
  Applications}.
\newblock In {\em {Principal Manifolds for Data Visualization and Dimension
  Reduction, volume 28,}}, pages 205--222. Springer-Verlag, 2007.

\bibitem{Schwarz1978}
G.~Schwarz.
\newblock {Estimating the dimension of a model}.
\newblock {\em Annals of Statistics}, 6:461--464, 1978.

\bibitem{Stock1999}
J.~Stock and M.~Stock.
\newblock {Quantitative stellar spectral classification}.
\newblock {\em Revista Mexican de Astronomia y Astrofisica}, 34:143--156, 1999.

\bibitem{Stock2008}
M.~J. Stock, J.~Stock, J.~Garcia, and N.~Sanchez.
\newblock {Quantitative Stellar Spectral Classification. II. Early Type Stars}.
\newblock {\em Rev.Mex.Astron.Astrofis. 38 (2002) 127-140}, May 2002.

\bibitem{Tenenbaum2000}
J.~B. Tenenbaum.
\newblock {A Global Geometric Framework for Nonlinear Dimensionality
  Reduction}.
\newblock {\em Science}, 290(5500):2319--2323, December 2000.

\bibitem{TB99}
M.~E. Tipping and C.~M. Bishop.
\newblock {Probabilistic principal component analysis}.
\newblock {\em Journal of the Royal Statistical Society, Ser. B},
  61(3):611--622, 1999.

\end{thebibliography}

\end{document}